\documentclass[letterpaper, 10 pt, conference]{ieeeconf}
\usepackage{array}

\usepackage[caption=false,font=small,labelfont=rm,textfont=rm]{subfig}

\usepackage{enumerate} 
\usepackage{textcomp}
\usepackage{stfloats}
\usepackage{url}
\usepackage{verbatim} 
\usepackage{graphicx}
\usepackage{subcaption} 
\usepackage[toc,page]{appendix}

\usepackage[english]{babel}
\usepackage{cite}
\usepackage{amsmath,amssymb,amsfonts}
\usepackage{caption}
\usepackage{subfig}

\usepackage{amsthm}\usepackage{bm}
\usepackage{graphicx}
\usepackage{balance}
\usepackage{float} 
\def\BibTeX{{\rm B\kern-.05em{\sc i\kern-.025em b}\kern-.08em
    T\kern-.1667em\lower.7ex\hbox{E}\kern-.125emX}}
\usepackage{subfloat}
\usepackage{tikz}
\usetikzlibrary{positioning, shapes.geometric,arrows}
\usepackage{mathtools}
\usepackage{makecell}
\usepackage{xparse}
\usepackage[ruled,linesnumbered]{algorithm2e}
\usepackage{algorithmicx}
\usepackage{algpseudocode}
\usepackage[colorlinks=true, allcolors=blue]{hyperref} 
\usepackage{comment}
\usepackage{soul}
\usepackage{cleveref}
\usepackage{flushend}

\def \st{self-triggered }

\newtheorem{remark}{Remark}
\theoremstyle{definition}
\newtheorem{definition}{Definition}

\newtheorem{myprob}{Problem}
\newtheorem{proposition}{\bf Proposition}
\newtheorem{thm}{\bf Theorem}
\newtheorem{myprof}{\bf Proof of Proposition}

\IEEEoverridecommandlockouts

\title{\bf Sleep When Everything Looks Fine:\\ Self-Triggered Monitoring for Signal Temporal Logic Tasks}
\author{Chuwei Wang, Xinyi Yu, Jianing Zhao, Lars Lindemann  and Xiang Yin
\thanks{This work was supported by  the national natural Science Foundation of China (62061136004, 62173226, 61833012).}
	\thanks{Chuwei Wang, Jianing Zhao and Xiang Yin are with Department of Automation and Key Laboratory of System Control and Information Processing, Shanghai Jiao Tong University, Shanghai 200240, China.
	e-mail: \tt\small \{wangchuwei, jnzhao, yinxiang\}@sjtu.edu.cn}
	\thanks{Xinyi Yu and Lars Lindemann are with Thomas Lord Department of Computer Science, University of Southern California, Los Angeles, CA 90089, USA.
	e-mail: \tt\small \{xinyi.yu12,llindema\}@usc.edu}
}

\begin{document}
\captionsetup[figure]{labelformat={default},labelsep=period,name={Fig.}}
\maketitle
\thispagestyle{empty}
\pagestyle{empty}
\setlength{\abovecaptionskip}{0pt}
\setlength{\belowcaptionskip}{0pt}
\setlength{\textfloatsep}{6pt}

\begin{abstract}
Online monitoring is a widely used technique in assessing if the performance of the system satisfies some desired requirements during run-time operation. Existing works on online monitoring usually assume that the monitor can acquire system information periodically at each time instant. However, such a periodic mechanism may be unnecessarily energy-consuming as it essentially requires to turn on sensors consistently. In this paper, we proposed a novel \emph{self-triggered}  mechanism for  model-based online monitoring of discrete-time dynamical system under specifications described by signal temporal logic (STL) formulae. 
Specifically, instead of sampling the system state at each time instant, a \st monitor can actively determine when the next system state is sampled in addition to its monitoring decision regarding the satisfaction of the task. We propose an effective algorithm for synthesizing such a \st monitor that can correctly evaluate a given STL formula on-the-fly  while maximizing the time interval between two observations.  We show that, compared with the standard online monitor with periodic information, 
the proposed \st monitor can significantly reduce observation burden while ensuring that no information of the STL formula is lost.  Case studies are provided to illustrate the proposed  monitoring mechanism.
\end{abstract}

\section{Introduction}

Autonomous systems, such as teams of ground robots or unmanned aerial vehicles, have found widespread applications in diverse areas, ranging from search and rescue missions to smart warehouses. The critical nature of these systems lies in their safety implications: any failure to accomplish their designated tasks could result in catastrophic consequences.  However, these systems are inherently susceptible to errors due to their intricate logics and complex dynamics. Therefore, monitoring the safety status of the system stands out as a central task during their operations \cite{yan2022distributed,bonnah2022runtime}. In particular, one should halt the system and initiate corrective actions once a failure is detected.

To express the formal  requirements of autonomous systems, temporal logic, including linear temporal logic (LTL) or signal temporal logic (STL), stands out as one of the most widely employed tools. This is because it offers a rich and structured framework for describing high-level tasks. Particularly, STL provides an effective tool for the quantitative evaluation of real-value signals in real-time settings \cite{donze2013efficient, donze2010robust}. Since the seminal work of Maler \cite{Maler2004}, STL has undergone extensive developments and applications in the analysis and control of numerous safety-critical systems \cite{ma2020sastl,srinivasan2020control,liu2022compositional,kurtz2022mixed,sun2022multi}.

Online monitoring is a widely adopted lightweight technique for  evaluating the correctness of the system  in real-time \cite{ho2014online,dokhanchi2014line,deshmukh2017robust}. In contrast to model checking, which necessitates offline enumeration of all possible behaviors, online monitoring involves a monitor that observes the \emph{partial signals} generated  up to the current instant. It then checks if the task can still  be successfully completed. Online monitoring techniques are broadly categorized into \emph{model-free} and \emph{model-based}, depending on the knowledge utilized by the monitor.
Specifically, in model-free online monitoring \cite{deshmukh2017robust, dokhanchi2014line}, correctness evaluation relies solely on the partial signal, neglecting the  model information of dynamic systems. More recently, there has been a surge of interest in model-based online monitoring \cite{qin2020clairvoyant, ma2021predictive,yoon2021predictive,lindemann2023conformal,momtaz2023predicate,Yu2024-auto}. This approach, leveraging the dynamic information of the system, enables the monitor to make more precise evaluations regarding the satisfaction of the task. 
For instance, in our recent work \cite{Yu2024-auto}, we proposed an integrated framework that combines offline computations of feasible sets based on the system model with online  decision-makings based on the partial signal.

Concerning the implementation of online monitoring, a fundamental question is how online information is obtained by the monitor. 
However, this crucial aspect has been neglected by existing works, as they implicitly presuppose that the monitor has complete access to the partial signal (state sequence) generated by the system. This assumption implies that the monitor can acquire system information periodically by, e.g., consistently activating sensors at each time instant. Nevertheless, in numerous applications, such a periodic information acquisition mechanism may prove to be unnecessary. Consistently activating sensors can be either energy-consuming or potentially prone to information leakage. This scenario is also applicable in various daily settings. For instance, a factory manager typically exercises greater caution when a machine approaches the safety boundary, and conversely, adopts a more relaxed stance when there are no hazardous surroundings within their field of view.

Motivated by the preceding discussion, this paper introduces a novel \emph{self-triggered} mechanism for online monitoring of discrete-time dynamical systems under STL specifications. In this approach, we adopt a model-based setting, assuming that the monitor possesses the dynamic model of the system. However, rather than sampling the system state periodically at each time instant, a self-triggered monitor can actively determine when to sample the next system state, alongside making monitoring decisions about the  feasibility of the STL task. As a result, the information acquisition module remains silent when critical information about the task status is not required. We present an effective algorithm for synthesizing such a self-triggered monitor capable of on-the-fly evaluation of a given (fragment)  STL formula. Specifically, the self-triggered monitor aims to maximize the time interval between two observations while ensuring the successful evaluation of the STL task.

Finally, it is worth mentioning that the concept of self-triggered as well as event-triggered mechanisms has found widespread applications in the literature, aiming to conserve sensor energy \cite{nowzari2012self}, alleviate computation burdens \cite{anta2010sample}, reduce communication bandwidth \cite{Brunner2019}, or enhance information security \cite{senejohnny2017jamming}. 
In \cite{lindemann2018event}, the authors applied event-triggered 
 mechanism to control synthesis of STL tasks. However, to  our knowledge, the utilization of self-triggered information acquisition mechanisms for the purpose of online monitoring of complex logic tasks has not been explored in the literature. 
 Our works is also related to online monitoring under parial observation. where state estimations are also involved; see, e.g., \cite{cairoli2021neural}. Still, no event-triggered mechanism is considered.

The remainder of the paper is organized as follows. 
Section \ref{sec-pre} introduce some basic preliminaries, and  the standard model-based online monitoring  problem is introduced in Section~\ref{sec-mom}.
In Section \ref{sec-pro}, we formally formulate the self-triggered monitoring problem that we solve in this paper. 
Then a self-triggered online monitoring algorithm as well as its correctness are investigated in Section \ref{sec-alg}.  In Section \ref{sec-case}, we present two case studies to illustrate the effectiveness of our algorithm. Finally, we conclude the paper in Section~\ref{sec-con}.

\section{Preliminaries}\label{sec-pre}
\subsection{System Model}
We consider a discrete-time dynamic  system of form
\begin{equation}
x_{t+1}=f(x_t,u_t),  
\label{system model}
\end{equation} 
where 
$x_t\in \mathcal{X}\subseteq \mathbb{R}^{n}$ and $u_t\in\mathcal{U}\subseteq\mathbb{R}^{m}$ are
the system state and the control input at time $t$, respectively,  and $f:\mathcal{X}\times\mathcal{U}\to\mathcal{X}$ is the dynamic function. 
We assume that the initial state is fixed as $x_0\in\mathcal{X}$. 
Given a sequence of control inputs $\mathbf{u}_{0:T-1}=u_{0}u_{1} \cdots u_{T-1}\in \mathcal{U}^{T}$, 
the resulting trajectory of the system is the state sequence  $\xi(x_0,\mathbf{u}_{0:T-1})=\mathbf{x}_{1:T}=x_1\cdots  x_T$ such that $x_{i+1}=f(x_i,u_i),i=0,\ldots,T-1$. 

\subsection{Signal Temporal Logic} 
The specifications of the system are described by STL formulae with bounded-time \cite{Maler2004}, 
whose  syntax  is as follows
\[
  \psi ::= \top \mid \pi^\mu \mid \neg \psi \mid \psi_1 \wedge \psi_2 \mid \psi_1 \mathbf{U}_{[a,b]} \psi_2,
\]
where 
$\top$ is the true predicate, 
$\pi^\mu$ is a predicate  whose truth value is determined by the sign of function $\mu: \mathbb{R}^{n} \to \mathbb{R}$, i.e.,  $\pi^\mu$  is true iff $\mu(x)> 0$. 
Notations $\neg$ and $\wedge$ are the standard Boolean operators ``negation" and ``conjunction", respectively, which can further induce ``disjunction" by $\psi_1\vee\psi_2:=\neg(\neg\psi_1\wedge\neg\psi_2)$ and ``implication" by $\psi_1\to\psi_2:=\neg\psi_1\wedge\psi_2$. 
Notation $\mathbf{U}_{[a,b]}$ is the temporal operator ``until", where $a,b\in\mathbb{R}_{\geq0}$ are two  time instants.  One can also induce  temporal operators ``eventually" and ``always" by $\mathbf{F}_{[a,b]}\psi:=\top\mathbf{U}_{[a,b]}\psi$ and $\mathbf{G}_{[a,b]}:=\neg\mathbf{F}_{[a,b]}\neg\psi$, respectively. 

Given a sequence $\mathbf{x}$, we use notation $(\mathbf{x},t_0)\models \psi$ to denote the satisfaction for STL formulae $\psi$ at time $t_0$. 
The semantics of STL are inductively defined as follows:
\begin{equation}
\begin{tabular}{lcl}
$(\mathbf{x},t)\models \pi^\mu$  &  iff  & $\mu(\mathbf{x}(t))>0$\\
$(\mathbf{x},t)\models \neg\psi$  &  iff  & $\neg((\mathbf{x},t)\models\psi)$\\
$(\mathbf{x},t)\models \psi_1 \wedge  \psi_2$  &  iff  & $(\mathbf{x},t)\models\psi_1\wedge (\mathbf{x},t)\models \psi_2$\\
$(\mathbf{x},t)\models\psi_1\mathbf{U}_{[a,b]}\psi_2$  &  iff  & 
$\exists t'\!\in\! [t+a,t+b]\!:\!(\mathbf{x},t')\!\models\! \psi_2 $\\
&& \!\text{and }$\forall t'' \in [t,t'']:(\mathbf{x},t')\models\psi_1$\notag
\end{tabular}
\end{equation}
We write $\mathbf{x}\models\psi$ whenever $(\mathbf{x},0)\models\psi$.
The readers are referred to \cite{Maler2004} for more details on the semantics of STL.

In the original semantics of $\mathbf{U}_{[a,b]}$, formula
$\psi_1$  needs to be satisfied within interval $[0,a]$. 
For the sake of convenience, 
we use new ``until" operator $\mathbf{U}'$ by slightly modifying the semantics as: $(\mathbf{x},t)\models \psi_1 \mathbf{U}'_{[a,b]} \psi_2$  iff
\[
[\exists t' \!\in\! [t+a, t+b] \!:\! (\mathbf{x},t') \!\models\! \psi_2]
\wedge
[\forall t'' \!\in\! [t+a, t']\!:\! (\mathbf{x},t'') \!\models\! \psi_1]
\]
Note that, we can express the original until operator $\mathbf{U}$ 
as 
$\psi_1 \mathbf{U}_{[a,b]} \psi_2=  \mathbf{G}_{[0, a]} \psi_1
\wedge \psi_1 \mathbf{U}'_{[a,b]} \psi_2$.  
Hereafter in the paper,  we will use this modified version of until operator, and rewrite it directly as  $\psi_1 \mathbf{U}_{[a,b]} \psi_2$ by omitting the superscript.

Finally, for predicate $\pi^\mu$, its satisfaction region  is 
$\mathcal{H}^{\mu}:=\{x\in \mathcal{X}\mid \mu(x) \!\geq\! 0\}$
with 
$\mathcal{H}^{\neg \varphi} = \mathcal{X}\setminus \mathcal{H}^{ \varphi}$ and 
$\mathcal{H}^{\varphi_1 \wedge \varphi_2} = \mathcal{H}^{\varphi_1} \cap \mathcal{H}^{\varphi_2}$.
Therefore, for any Boolean formula  $\varphi$, we can express its requirement equivalently by  $x \in \mathcal{H}^\varphi$.

\subsection{Fragment of STL Formulae}
In this work, we consider a fragment of STL formulae such that the overall task is expressed as the conjunction of $N$ sub-formulae without nested temporal operators. 
Formally, we consider an STL formula of form: 
\begin{equation}
    \Phi= \bigwedge_{i=1,\dots,N} \Phi_{i},   \label{STL form}
\end{equation}
where each $\Phi_i$ is either
\[
\text{(i) }
\mathbf{G}_{[a_i,b_i]} x\!\in\! \mathcal{H}_i
\quad \text{  or } \quad \text{(ii) }
x\!\in\!\mathcal{H}_i^1 \mathbf{U}_{[a_i,b_i]} x\!\in\! \mathcal{H}_i^2.
\]
Let  $\mathcal{I} =\{1,\dots ,N\} $ be the index set of sub-formulae.  
We assume that the indices are ordered based on the beginning instant of each sub-formulae, i.e.,   $a_1\leq  \cdots \leq a_N$. 
For each time instant $t$, 
we denote by 
$\mathcal{I}_t = \{ i  \mid a_i \!\leq\! t \!\leq\! b_i \} \subseteq \mathcal{I}$
the index set of formulae that are effective at time $t$.  
Similarly,  we define the index sets of  sub-formulae effective before and after  $t$  by 
$\mathcal{I}_{<t} \!=\! \{ i  \mid t \!>\! b_i  \}$ and $\mathcal{I}_{>t} \!= \!\{ i  \mid  t\!<\! a_i \}$, respectively.
For each sub-formulae $i \in \mathcal{I}$, we denote by $\mathbf{O}_i \in \{\mathbf{G},  \mathbf{U}\}$  
the unique temporal operator in $\Phi_i$.  
We define $\mathcal{I}_t^\mathbf{U}=\{i\in \mathcal{I}_t\mid \mathbf{O}_i=\mathbf{U} \}$; 
the same for $\mathcal{I}_t^\mathbf{G}$.

\subsection{Remaining Formulae and Feasible Sets}

As the system evolves, some sub-formulae may be satisfied or not effective anymore. 
Throughout the paper, we will use notation $I\subseteq \mathcal{I}$ to denote the index set for those sub-formulae remaining unsatisfied. 
If  $I\cap \mathcal{I}_{<t}=\emptyset$, we define the 
\emph{$I$-remaining formula} at instant $t$ in the form of
\begin{equation}
    \hat{\Phi}_t^I= \underset{i\in I\cap\mathcal{I}_t}{\bigwedge}\Phi_i^{[t,b_i]} \wedge \underset{i\in \mathcal{I}_{>t}}{\bigwedge}\Phi_i, 
\end{equation} 
where $\Phi_i^{[t,b_i]}$ is attained from $\Phi_i^{[a_i,b_i]}$ by replacing the starting instant of the temporal operator from $a_i$ to $t$.
Then the \emph{$I$-remaining feasible set} at instant $t$, denoted by ${X}_t^I$, is 
defined as the set of states from which  the system can possibly satisfy the $I$-remaining formula, i.e., 
\begin{equation}\label{feasible set}
    {X}_t^I = \left\{
	x_t \in \mathcal{X} \,\middle\vert\, 
	\begin{array}{cc}
		\exists \ \mathbf{u}_{t:T-1} \in \mathcal{U}^{T-t} \\
		\text{ s.t. } x_t \xi_f(x_{t}, \mathbf{u}_{t:T-1}) \models \hat{\Phi}_{t}^I
	\end{array}  
	\right\}.
\end{equation}

Otherwise, 
if $I\cap \mathcal{I}_{<t}\neq \emptyset$, 
it means that there are some sub-formulae that should have been satisfied. In this scenario, we define ${X}_t^I=\emptyset$ since the overall task is already failed. 

\section{Model-Based Online Monitoring}
\label{sec-mom}
\subsection{Standard Online Monitor with Periodic Sampling}
In the context of online monitoring, 
the system state is observed by a \emph{monitor} that determines the satisfaction of the STL task dynamically online based on the trajectory generated by the system. 
Specifically, given an STL formula $\Phi$ with  time horizon $T$ and a (partial) signal  $\mathbf{x}_{0:t}=x_0x_1\dots x_t$ (also called \emph{prefix}) up to time instant $t<T$, 
we say    $\mathbf{x}_{0:t}$ is  
\begin{itemize}
    \item 
    \emph{violated} if   $\mathbf{x}_{0:t}\xi(x_t,\mathbf{u}_{t:T-1})\!\not \models\!\Phi$ 
    for any  $\mathbf{u}_{t:T-1}$; 
    \item 
    \emph{feasible} if  $\mathbf{x}_{0:t}\xi(x_t,\mathbf{u}_{t:T-1})\! \models \!\Phi$ for some  $\mathbf{u}_{t:T-1}$. 
\end{itemize}
Then an online monitor is a function
\begin{equation}
    \mathcal{M}: \mathcal{X}^*\to\{0,1\}
\end{equation}
such that, for any prefix $\mathbf{x}_{0:t}$, we have 
\begin{equation}
\mathcal{M}(\mathbf{x}_{0:t})=1 \quad\Leftrightarrow\quad  \mathbf{x}_{0:t}\text{ is violated},
\end{equation}
where $\mathcal{X}^*$ denotes the set of all finite sequence over $\mathcal{X}$.

The above defined monitor is referred to as the \emph{periodic monitor} hereafter 
since it needs to acquire  state information  $x_t$ periodically at each time instant. 
In \cite{Yu2024-auto}, an effective algorithm has been proposed to synthesize  a model-based periodic monitors  involving both offline computations and online execution. 
Specifically, for the offline stage, one first  use the model information to pre-compute  all possible $I$-remaining feasible sets that may incur for each time instant. Then for the online execution, one simply tracks the index set $I$ for the remaining formulae, 
and the monitoring decision can be made by checking whether or not $x_t\in {X}_t^I$.

\subsection{Motivating Example for Aperiodic Sampling}
Before formally formulating the problem, let us consider a motivating example. We consider a nonholonomic unicycle mobile  robot working in a factory. It needs to arrive at the charging station in 10 minutes while avoiding the conveyors on its path. The charging station is marked by the green region in Fig.~\ref{fig:example}.  

Clearly, a periodic monitor can successful monitor the task by continuously checking whether or not the robot hits the obstacle or reaches the charging station on time. 
However, it may require unnecessary observations. 
For example,  consider a possible trajectory shown in Fig.~\ref{fig:example} and two locations points $P1$, $P2$ on it. 
The sector region around each point is the area the robot can reach within 2 minutes according to its dynamic. 
If the robot is at $P1$, then it may hit a conveyor in 2 minutes. Meanwhile, when the robot is at $P2$, no collision is possible in 2 minutes.  Therefore, to reduce the cost of sensor deployment, the monitor can ``sleep'' for at least 2 minutes, instead of constantly collecting irrelevant data.
\begin{figure}[t]
	\centering
{\includegraphics[width=0.9\linewidth]{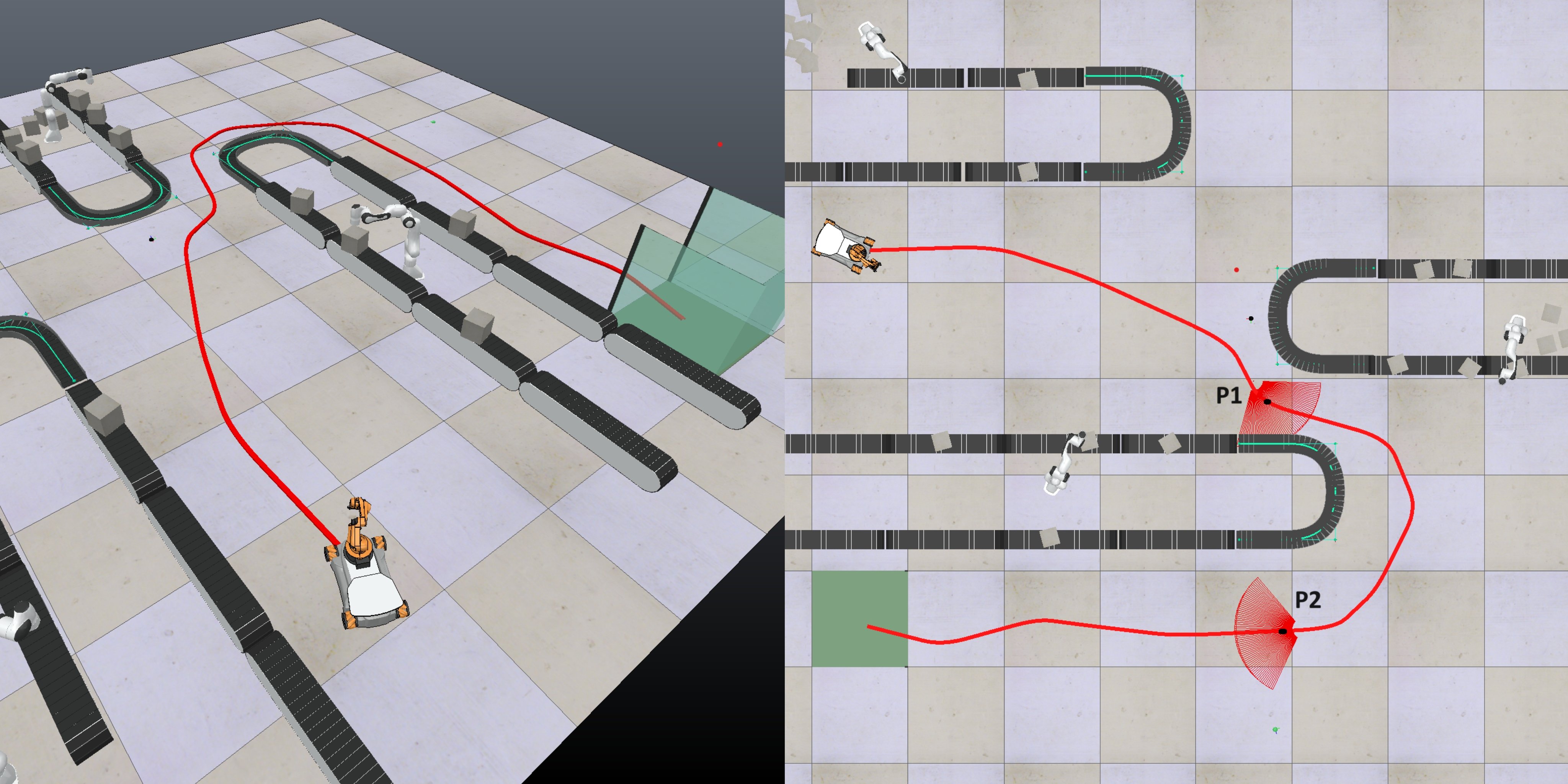}}
\caption{Motivating example of robot charging mission.}
\label{fig:example}
\end{figure}

\section{Self-Triggered Monitoring Mechanism}\label{sec-pro}
In this section, we formally present the proposed \st monitoring mechanism. 
Compared with the periodic monitor that passively receives state information at each time instant, a \st monitor will \emph{actively} determine the next time instant for acquiring state information. 
Specifically, a \st monitor works as follows:
\begin{itemize}
    \item 
    At each decision instant $t$ (which is determined at the previous decision round), 
    it observes the current state $x_t$ by, e.g., turning on sensors,  and makes a monitoring decision $0$ or $1$ regarding the violation of the STL task;  
    \item 
    Simultaneously, together with the binary monitoring decision, it also determines a number $\tau$ specifying after how many time instants the next observation will be made; 
    \item 
    Then the system evolves according to its own dynamic ``silently" to the monitor until time instant $t+\tau$, at which the above two steps are repeated. 
\end{itemize}

To formalize the above process, we define an \emph{observation history} as a sequence of states with time stamps of form 
\begin{equation}
\hbar \coloneqq (x_0,\tau_0)(x'_1,\tau_1)\ldots (x'_{n-1},\tau_{n-1})x'_n\in (\mathcal{X}\times \mathbb{N})^*\mathcal{X}. 
\end{equation}
Specifically,  each $x'_i$ denotes the state at the $i$-th observation and 
$\tau_i$ denotes the time between observations $x'_{i}$ and   $x'_{i+1}$. 
Note that, we use prime in $x'_i$ since it is not the system state at time instant $i$. Then we denote by  $\mathbb{H}:=(\mathcal{X}\times \mathbb{N})^*\mathcal{X}$ the set of all observation histories.  

With the notion of observation histories, now we are ready to formally define the \st monitors. 

\begin{definition}[\bf Self-Triggered Monitors]
Let $T_{max}<T$ be a number specifying the maximum time the monitor allowed to stay silent.  A \st monitor is a function 
\begin{equation}
    \mathcal{M}: \mathbb{H} \to \{0,1\} \times \{1,\ldots ,T_{max}\}
\end{equation} 
that makes decisions  based on the  observation histories.
Specifically, for any  observation history $\hbar \in \mathbb{H}$, 
we have  $\mathcal{M}(\hbar) \coloneqq (\mathcal{M}_D(\hbar), \mathcal{M}_\tau(\hbar))$,  where
\begin{itemize}
    \item 
    $\mathcal{M}_D(\hbar) \in \{0,1\}$ is the  binary monitoring decision 
    with ``0'' denoting \emph{feasible}  and ``1'' denoting \emph{violated}; and 
    \item 
    $\mathcal{M}_\tau(\hbar) \in  \{1,\ldots ,T_{max}\}$ is the trigger information  indicating the time interval to the next observation.
\end{itemize}
Once $\mathcal{M}_D(\hbar)$ returns $1$, the monitoring process terminates.
\end{definition}

Clearly, a \st monitor reduces to the standard periodic monitor when $\mathcal{M}_\tau(\hbar)$ is set to be a single time unit constantly. In this case, the monitor will precisely observe the entire trajectory generated by the system. However, in general, the monitor can only obtain  partial information from the underlying trajectory, which is captured by the following definition.

\begin{definition}[\bf Induced Observation Histories]
Let $\mathcal{M}$ be a \st monitor and  $\mathbf{x}_{0:t} = x_0 x_1 \dots x_t$ be a state sequence generated by   system \eqref{system model}. 
The  \emph{observation history} of $\mathbf{x}_{0:t} $  \emph{induced} by $\mathcal{M}$, denoted by $\hbar_\mathcal{M}(\mathbf{x}_{0:t})$, is defined as the unique sequence
\begin{equation}
    \hbar_\mathcal{M}(\mathbf{x}_{0:t}) =(x_0,\tau_0)(x'_1,\tau_1)\ldots (x'_{n-1},\tau_{n-1})x'_n \in \mathbb{H}
\end{equation}
such that
\begin{enumerate}[(i)]
    \item 
    $\forall i \geq 1:  x'_i = x_{\tau_0 + \dots +\tau_{i-1}}$; 
    \item 
    $\forall i \geq 0: \mathcal{M}_\tau((x_0,\tau_0)(x'_1,\tau_1)\cdots (x'_{i-1},\tau_{i-1}) x'_i) \!=\!\tau_{i}$;
    \item 
    $\sum_{i=0}^{n-1}\tau_i   \leq t < \sum_{i=0}^{n-1}\tau_i  + \mathcal{M}_\tau(\hbar_\mathcal{M}(\mathbf{x}_{0:t}))$.
\end{enumerate} 
\end{definition}

The intuitions of the above definition are as follows. 
The first condition says that states in the induced observation history are sampled from the original state sequence at those decision  time stamps. 
The second  condition says that the  time stamp for the $i$-th observation is determined by  $\mathcal{M}$ based on the previous observation history. 
Finally, the last condition says that there is no more observation after time instant $t$.
 
Now, we formulate the \st monitor synthesis problem that we solve in this paper as follows. 

\begin{myprob}[\bf Self-Triggered Monitor Synthesis Problem]
Given a dynamic system of form \eqref{system model} and an STL formula $\Phi$ as in \eqref{STL form}, design 
a self-triggered online monitor $\mathcal{M}: \mathbb{H} \to \{0,1\} \times \{1,\ldots ,T_{max}\}$ such that for any prefix signal $\mathbf{x}_{0:t}$,  we have $\mathcal{M}_D(\hbar_\mathcal{M}(\mathbf{x}_{0:t})) = 1$ if and only if $\mathbf{x}_{0:t}$ is violated.

\end{myprob}

We conclude this section by making some remarks regarding the above problem formulation.
\begin{remark}
Under the \st mechanism, only  partial signals are available to the monitor. Therefore, it is possible that the system has two distinct trajectories having the same observation. 
Therefore,  Problem~1 essentially requires that the monitor can always distinguish two trajectories such that one is feasible and the other is violated. 
Furthermore, since the monitor does not receive any information until the next trigger time instant,  any signals within this ``silent" interval should also result in the same task status. 
\end{remark}

\begin{remark}
In Problem~1, we do not explicitly put any optimization requirement for the trigger interval $\mathcal{M}_\tau(\hbar)$ for each decision instant. 
In principle, one may want $\mathcal{M}_\tau(\hbar)$ to be as large as possible for the purpose of energy saving. 
Hereafter, our solution algorithm will follow a greedy strategy, i.e., at each decision instant, the monitor tries to maximize the current trigger interval 
as long as the ambiguity requirement is satisfied without further considering its future effect. 
\end{remark}

\section{Self-Triggered Online Monitoring Algorithm}\label{sec-alg}
In this section, we present our solution for synthesizing \st monitor. 
Specifically, we first introduce the notion of belief states and their evolution. 
Then an effective online monitoring algorithm is proposed based on the belief states. 
Finally, we prove the correctness of the algorithm. 
\subsection{Belief States and Predictions}\label{sec-algA}

For online monitoring with periodic information, one can always maintain a precise information $(x_t,I_t)$ regarding the current status of the system, where $x_t$ is the current state and $I_t\subseteq \mathcal{I}$  is the index set for the remaining sub-formulae representing the progress of the task.  
However, for the case of \st monitoring,  one of the major challenges is that one does not have this precise information due to unobservable transitions between two sampling instants. Hereafter, we will refer to tuple $(x_t,I_t)$ as an \emph{augmented state} at time instant $t$. Instead of maintaining a precise augmented state at each instant, the system can only maintain a \emph{belief state} to estimate the current status of the system. 
Note that, for technical purpose, 
in each augmented state $(x_t,I_t)$,  the progress of the task $I_t$ does not take the effect of current state $x_t$, 
which will be carried on for the next instant. 

\begin{definition}[\bf Belief States] 
A \emph{belief state}   is a set of augmented states. 
We denote by  $\mathbb{B} \coloneqq  2^{\mathcal{X}\times \mathcal{I}}$ the set of all possible belief states. 
\end{definition}

Given a belief state, each augmented state $(x_t,I_t)$ in it is an \emph{explanation} of the possible current status of the system.
Therefore, a belief state essentially captures all possible explanations based on the partial information. 
Now suppose that the belief state the monitor holds at instant $t$ is $\bm{b}$, then for the next instant $t+1$, 
it can update its belief state as follows: 
\begin{itemize}
    \item 
    If there is no observation at instant $t+1$, then it can make an ``open-loop prediction" based on the dynamic of the system as well as the semantics of STL formulae; 
    \item 
    If a new state is observed, i.e., $t+1$ is the triggered instant decided in the previous round, 
    then it can further use the state observation to refine the belief. 
\end{itemize}

Next, we elaborate on more details about how belief states are updated. According to the semantics of STL formulae, the index set can be updated once a new state is reached. 

\begin{definition}[\bf Index Updates]
For each time instant $t$, 
we define the \emph{index update function}, 
denoted as
$\textsf{update}_t: \mathcal{I}\times \mathcal{X}\to \mathcal{I}$, 
by:  for any $I\subseteq \mathcal{I},x\in \mathcal{X}$, we have 
\begin{equation}
\label{update}
    \!\!\!\!\!\textsf{update}_t(I,x) 
  \!=\! 
  \left\{
	i \!\in\! I \,\middle\vert\, \!\!\!
	\begin{array}{cc}
	[i\in \mathcal{I}_{t}^\mathbf{U}\wedge x \notin  \mathcal{H}_i^1 \cap \mathcal{H}_i^2] \\
	\vee [i\in \mathcal{I}_{t}^\mathbf{G}\wedge t \neq b_i]\vee[i\notin \mathcal{I}_{t}]
	\end{array}  
 \!\!\!
	\right\} \!, \!\!\!\!\!\!
\end{equation}
where,  for each $i\!\in\! \mathcal{I}_{t}^\mathbf{U}$,  we have $\Phi_i\!=\! x\!\in\!\mathcal{H}_i^1 \mathbf{U}_{[a_i,b_i]} x\!\in\! \mathcal{H}_i^2$. 
\end{definition}

Intuitively, suppose that $I$ is the remaining index set at instant $t$ and $x$ is the state reached at instant $t$. 
Then $\textsf{update}_t(I,x)$ essentially excludes the index set for those sub-formulae with temporal operator $\mathbf{U}$ that have already been satisfied by reaching state $x$ based on the current index set $I$. Meanwhile, those sub-formulae with temporal operator $\mathbf{G}$ expired will not be contained anymore.

Now, by combining the dynamic of the system with the update function of the index set, 
we can define the ``dynamic" over the augmented state space. 

\begin{definition}[\bf Successor Augmented States]
Let $(x,I)$ and $(x',I')$ be two augmented states at  instants $t$ and $t+1$, respectively. 
We say $(x',I')$ is a \emph{successor augmented state} of $(x,I)$ from instants $t$ to $t+1$ if 
\begin{enumerate}[(i)]
  \item 
  there exists $u\in \mathcal{U}$ such that $ x'=f(x,u)$; and 
  \item 
  $I'= \textsf{update}_t(I,x)$.
\end{enumerate} 
We denote by $\textsf{succ}_t(x,I) \in \mathbb{B}$ the set of all \emph{successor augmented states} of   $(x,I)$ from instants $t$ to $t+1$.
\end{definition}

Based on the  ``dynamic'' of   augmented states, we can define how belief states are evolved. 
First, without observations between two trigger instants, 
the belief states can be \emph{predicted} in an open-loop fashion as follow. 

\begin{definition}[\bf Belief Predictions] 
We define the (one-step) \emph{belief prediction function} from instants $t$ to $t+1$, denoted as $\textsf{pred}_{t}^{t+1}\!:\! \mathbb{B}\!\to\! \mathbb{B}$,  by: for any $\bm{b}\!\in\! \mathbb{B}$, we have 
\begin{equation}
    \textsf{pred}_{t}^{t+1}(\bm{b})= \underset{(x,I)\in \bm{b}}{\cup} \textsf{succ}_t(x,I).
\end{equation}
The multi-step belief prediction function is defined recursively by: for any $\bm{b}\!\in\! \mathbb{B}$ and $k\geq 1$, we have
\begin{equation}
    \textsf{pred}_{t}^{t+k}(\bm{b})=  
    \textsf{pred}_{t+k-1}^{t+k}(\cdots (\textsf{pred}_{t+1}^{t+2}(\textsf{pred}_{t}^{t+1}(\bm{b})))).
\end{equation}
\end{definition}

Once   observations  are made at those triggered instants, the belief states can be \emph{refined} as follows. 
\begin{definition}[\bf Belief Refinements] 
We define the \emph{belief refinement function}, denoted as $\textsf{refine}:  \mathbb{B}\times \mathcal{X} \to  \mathbb{B}$,  by: for any $\bm{b}\in \mathbb{B},x\in \mathcal{X}$, we have 
\begin{equation}
    \textsf{refine}(\bm{b},x)= 
    \{ (x',I')\!\in\! \bm{b} \mid x'\!=\!x   \}
    = \bm{b}\cap (\{x \}  \!\times\!\mathcal{I}).
\end{equation} 
\end{definition}

Intuitively, the belief refinement function simply  restricts a belief state to a smaller set that is consistent with the current state observed. 
Note that, compared with the belief prediction function, which is time-dependent,  the belief refinement function is time-independent. 

\subsection{Main Monitoring Algorithm}

First, we introduce two key notions that will be used then present our overall algorithm.
\begin{definition}[\bf Safe and Determined Beliefs]
Let $\bm{b}\in \mathbb{B}$ be a belief state  at instant $t$. 
We say belief state $\bm{b}$ is 
\begin{itemize}
    \item 
    \emph{safe} (w.r.t\ instant $t$) if  
    for any augmented state  $(x,I)\in \bm{b}$, 
    we have  $x\!\in\! X_{t}^{I}$;  and 
    \item 
    \emph{determined}  (w.r.t\ instant $t$) 
    if for any two augmented states $(x,I),(x,I')\in \bm{b}$ with the same state component, 
    we have $\textsf{update}_t(I,x) = \textsf{update}_t(I',x) $.
\end{itemize} 
\end{definition}

The intuitions for the above definitions are as follows. 
For each augmented state $(x,I)$, if $x\not\in X_{t}^{I}$, then it means that starting from state $x$, remaining sub-formulae in $I$ is no longer feasible. 
Therefore, this unsafe circumstance should be avoided for any augmented states in a belief state.   
On the other hand, a belief state is \emph{determined} means that 
there will be no ambiguity regarding the task progress once the actual current state is known. Recall that, in each augmented state $(x,I)$, the index set $I$ has not yet taken the effect of $x$ into account. 
Therefore, the actual index set for sub-formulae remaining is $\textsf{update}_t(I,x)$ rather that $I$ when $x$ is observed.  

Now, let us discuss how the online monitor decides when to trigger the next observation. 
Suppose that at time instant $t>0$,  after making an observation, 
the monitor knows that the current belief state  is $\bm{b}_{t}$, 
and it wants to determine an integer $\tau\in \{1,\dots,T_{max}\}$ as the time interval before the next observation. Such integer $\tau$ should satisfy the following two requirements: 

\SetKwRepeat{Do}{do}{while}
\begin{algorithm}[t]
    \SetAlgorithmName{Procedure}{}{}
	\caption{\texttt{Trigger-Time}}\label{procedure}
	\KwIn{belief state $\bm{b}$, current instant $t$}
	\KwOut{ trigger time $\tau$}
	$k\gets 1$; $\tau^\star\gets 1$; $\boldsymbol{\hat{b}}\gets \bm{b}$\\ 
    \While{$k\leq T_{max}$} 
        { 
        
        $\boldsymbol{\hat{b}}\gets  \textsf{pred}_{t+k-1}^{t+k}( \boldsymbol{\hat{b}} )$\\ 
        \If{$\boldsymbol{\hat{b}}$ is determined w.r.t.\ $t+k$}
        {
        $\tau^\star\gets k$; 
        }
        \eIf{$\boldsymbol{\hat{b}}$  is  safe w.r.t.\ $t+k$}{$k \gets k+1$ }
        {
        \textbf{return} $\tau^\star$
        } 
        }
\end{algorithm}

\begin{algorithm}[t]
\setcounter{algocf}{0}
  \KwIn{Feasible Sets}
  \KwOut{ decision $(\mathcal{M}_{\tau},\mathcal{M}_{D})$ for each instant}
  $t\gets 0$; $I\gets \mathcal{I}$; $\boldsymbol{\hat{b}}\gets\{ (x_0,I) \}$\\
  \While{a new state $x_t$ is observed}
  {
  $ \bm{b}\gets \textsf{refine}(\boldsymbol{\hat{b}},x_t)$\\
   \eIf{$\bm{b}$ is safe}
    {
    $(\mathcal{M}_{\tau},\mathcal{M}_{D}) \gets (\texttt{Trigger-Time}(\bm{b},t),0)$\\ 
    $\boldsymbol{\hat{b}}\gets \textsf{pred}_t^{t+\mathcal{M}_{\tau}}( \bm{b} )$\\
    wait for $\mathcal{M}_{\tau}$ instants and $t \gets t+\mathcal{M}_{\tau}$
    }
    { 
     $ \mathcal{M}_{D}\gets 1$ and \textbf{return} ``violated''
    } 
    }
\label{algortthm1}
\caption{Self-Triggered Monitor}\label{alg}
\end{algorithm}

\begin{itemize}
    \item 
\emph{Prediction Determinacy: }
Since there is no observation until time instant $t+\tau$, the monitor can only predict the belief state for time instant $t+\tau$ as     $\boldsymbol{\hat{b}}_{t+\tau}=\textsf{pred}_t^{t+\tau}(\bm{b}_t)$.  
Once the monitor takes observation at $t+\tau$ and observes a new state $x_{t+\tau}$, the belief state will shrink to $\bm{b}_{t+\tau}= \textsf{refine}(\boldsymbol{\hat{b}}_{t+\tau},x_{t+\tau})$.
Therefore, in order to capture the exact task progress of the system after making the observation, $\boldsymbol{\hat{b}}_{t+\tau}$ has to be determined. 
\item 
\emph{Prediction Safety: }
Furthermore, once  the  predicted  belief state is not safe, 
the monitor must \emph{immediately} make an observation to resolve ambiguity in the feasibility of the task.  Otherwise, it is possible that the system  has already become infeasible 
but the monitor does not issue an alarm on time. 
\end{itemize} 
Combining the above two requirements together, based on the current information $\bm{b}_t$, the largest value for the trigger interval the monitor can select is 
\begin{equation}
\label{tau}
    \tau^\star = \max\{   \tau \leq \tau_1  :    \textsf{pred}_t^{t+\tau}(\bm{b}_t) \text{ is determined}  \},
\end{equation}
where 
\begin{equation}
    \tau_1 = \min\{   \tau \leq T_{max}  :    \textsf{pred}_t^{t+\tau}(\bm{b}_t) \text{ is not safe}  \}.
\end{equation} 
This value $\tau^\star$ can be selected according to  procedure \texttt{Trigger-Time} 
by iteratively computing the predicted beliefs and checking their safety and determinacy.

Finally, we present our complete \st online monitoring algorithm for STL tasks   in Algorithm 1. We start from the initial instant $t=0$ and $I$ contains all the indexes in $\mathcal{I}$. At each iteration, we observe the system state and refine the belief state. Then we check if the refined belief state is safe or not. 
 If it is safe, then we continue the monitoring process and   employ procedure \texttt{Trigger-Time} to compute the next observation instant as well as the predicted belief states.
 If the refined belief state is not safe, 
 the monitor will issue an alarm and terminate the process.

\subsection{Properties of the Proposed Algorithm}
In this subsection, we prove the correctness of Algorithm~1. 
For any prefix state sequence $\mathbf{x}_{0:t}$, we define its induced augmented state sequence by 
\[
\zeta(\mathbf{x}_{0:t})= (x_0,I_0)(x_1,I_1)\ldots(x_t,I_t), 
\]
where for any $k=1,\dots, t$,  
$I_k$ is the index set for those sub-formulae who have not yet been satisfied by $x_0\cdots x_{k-1}$, i.e.,  $I_k=\{ i\in \mathcal{I} \wedge i\notin \mathcal{I}^{\mathbf{G}}_{<k}\mid x_0\cdots x_{k-1}\not\models \Phi_i   \}$.  
We denote by $\textsf{last}(\zeta(\mathbf{x}_{0:t}) )=(x_t,I_t)$ the last augmented state in the induced sequence. Then for any observation history $\hbar=(x_0,\tau_0)(x'_1,\tau_1)\ldots (x'_{n-1},\tau_{n-1})x'_n$, 
we define 
\[
\mathcal{M}^{-1}_{\bm{b}}(\hbar)=\{\zeta(\mathbf{x}_{0:t}): \hbar_{\mathcal{M}}( \mathbf{x}_{0:t} ) =\hbar \}
\]
as the set of augmented state sequence consistent with the observation under monitor $\mathcal{M}$, where  $t=\sum_{i=0}^{n-1}\tau_i$.
We  define
$\textsf{last}(\mathcal{M}_{\bm{b}}^{-1}(\hbar))=\{  \textsf{last}(\zeta(\mathbf{x}_{0:t}) )\mid    \zeta(\mathbf{x}_{0:t})\in \mathcal{M}^{-1}_{\bm{b}}(\hbar)\}  $. 

Note that, the above $\textsf{last}(\mathcal{M}_{\bm{b}}^{-1}(\hbar))$ is defined based on the entire observation history. However, Algorithm~1 computes belief in a recursive manner. 
Formally, still let $ \hbar  =(x_0,\tau_0)(x'_1,\tau_1)\ldots (x'_{n-1},\tau_{n-1})x'_n$ be an observation history. 
We define a belief sequence $B(\hbar )=\bm{\hat{b}}_0 \bm{b}_0\bm{\hat{b}}_1 \bm{b}_1\cdots\bm{\hat{b}}_n \bm{b}_n $ recursively  by  
\begin{itemize}
    \item 
    $\boldsymbol{\hat{b}}_0 =\{ (x_0,\mathcal{I}) \}$; 
    \item 
    $\forall k\geq 0,\bm{b}_k=\textsf{refine}( \bm{\hat{b}}_k,x'_k )$; 
    \item 
    $\forall k\geq 0,\bm{\hat{b}}_{k+1} = \textsf{pred}_{T_{k-1}}^{T_{k}}( \bm{b}_{k} )$.
\end{itemize}
where $T_k= \sum_{i=0}^{k}\tau_i$ with $T_{-1}=0$.   
According to Algorithm~1, 
$ \bm{b}_n$ is the belief state maintained by the monitor following observation history $\hbar$. 

The following result states that the belief state maintained by the monitor is indeed 
the set of all possible augmented states consistent with the observation.

\begin{proposition}\label{prop}
Let $\hbar=(x_0,\tau_0)(x'_1,\tau_1)\ldots (x'_{n-1},\tau_{n-1})x'_n$ be an observation history under a given monitor $\mathcal{M}$ and 
$B(\hbar )=\bm{\hat{b}}_0 \bm{b}_0\bm{\hat{b}}_1 \bm{b}_1\cdots\bm{\hat{b}}_n \bm{b}_n $ be its induced belief state sequence. 
Then we have  $\bm{b}_n=\textsf{last}(\mathcal{M}_{\bm{b}}^{-1}(\hbar))$. 
\end{proposition}

With the help of Proposition~\ref{prop}, now we can formally establish the correctness of Algorithm~1.

\begin{thm}
Algorithm \ref{alg} correctly solves Problem~1. 
That is, for any  $\mathbf{x}_{0:t}$ generated by system \eqref{system model},  
we have 
\begin{equation}
    \mathcal{M}_D(\hbar_\mathcal{M}(\mathbf{x}_{0:t})) = 1 \Leftrightarrow \mathbf{x}_{0:t} \text{ is violated}.
\end{equation} 
\end{thm}

\begin{proof} 
Let $\mathbf{x}_{0:t}$ be the actual prefix signal, 
$\zeta(\mathbf{x}_{0:t})= (x_0,I_0)(x_1,I_1)\ldots(x_t,I_t)$ be its induced augmented state sequence
and $\hbar=\hbar_\mathcal{M}(\mathbf{x}_{0:t})$ be the observation history of the monitor. 
Note that, 
the actual augmented state sequence  $\zeta(\mathbf{x}_{0:t})$ is always contained in 
$\mathcal{M}_{\bm{b}}^{-1}(\hbar)$, 
which also means that $(x_t,I_t) \in \textsf{last}(\mathcal{M}_{\bm{b}}^{-1}(\hbar))$. 
Furthermore,  according to Proposition~\ref{prop}, the belief state of the monitor after refinement is $\bm{b}=\textsf{last}(\mathcal{M}_{\bm{b}}^{-1}(\hbar))$. 

Now, let us argue the correctness of Algorithm~1 as follows.  
To show the ``$\Leftarrow$" direction, suppose that $\mathbf{x}_{0:t}$ is a violated sequence. 
Then  we know that some sub-formulae not yet satisfied must be no longer feasible from $x_t$, 
which means that $x_t\notin X_t^{I_t}$. 
This immediately implies that $\bm{b}$ is not safe, i.e., the monitor will issue decision $ \mathcal{M}_D=1$. 
For the ``$\Rightarrow$" direction, suppose that $\mathcal{M}_D(\hbar_\mathcal{M}(\mathbf{x}_{0:t}))=1$.  
This implies that $\bm{b}$ is not safe, i.e., 
there exists $(x,I)\in \bm{b}$ such that $x\notin X_t^I$. If $\bm{b}$ is a singleton, then $\mathbf{x}_{0:t}$ is a violated sequence. Otherwise, for all $(x_t,I)\in \bm{b}$, we have $x_t\notin X_t^I$. Based on procedure \texttt{Trigger-Time}, $\bm{\hat{b}}_n $ is determined, $x_t$ must be in the region consistent for index update thus the violation is only concerned with the remaining formulae after update. Thus we prove that $\mathbf{x}_{0:t}$ is a violated sequence.
\end{proof}

\begin{figure}[htbp]\centering
	\subfloat[Two possible altitude trajectories of the drone.]
	{ 
    \centering
    		\includegraphics[scale=0.23]{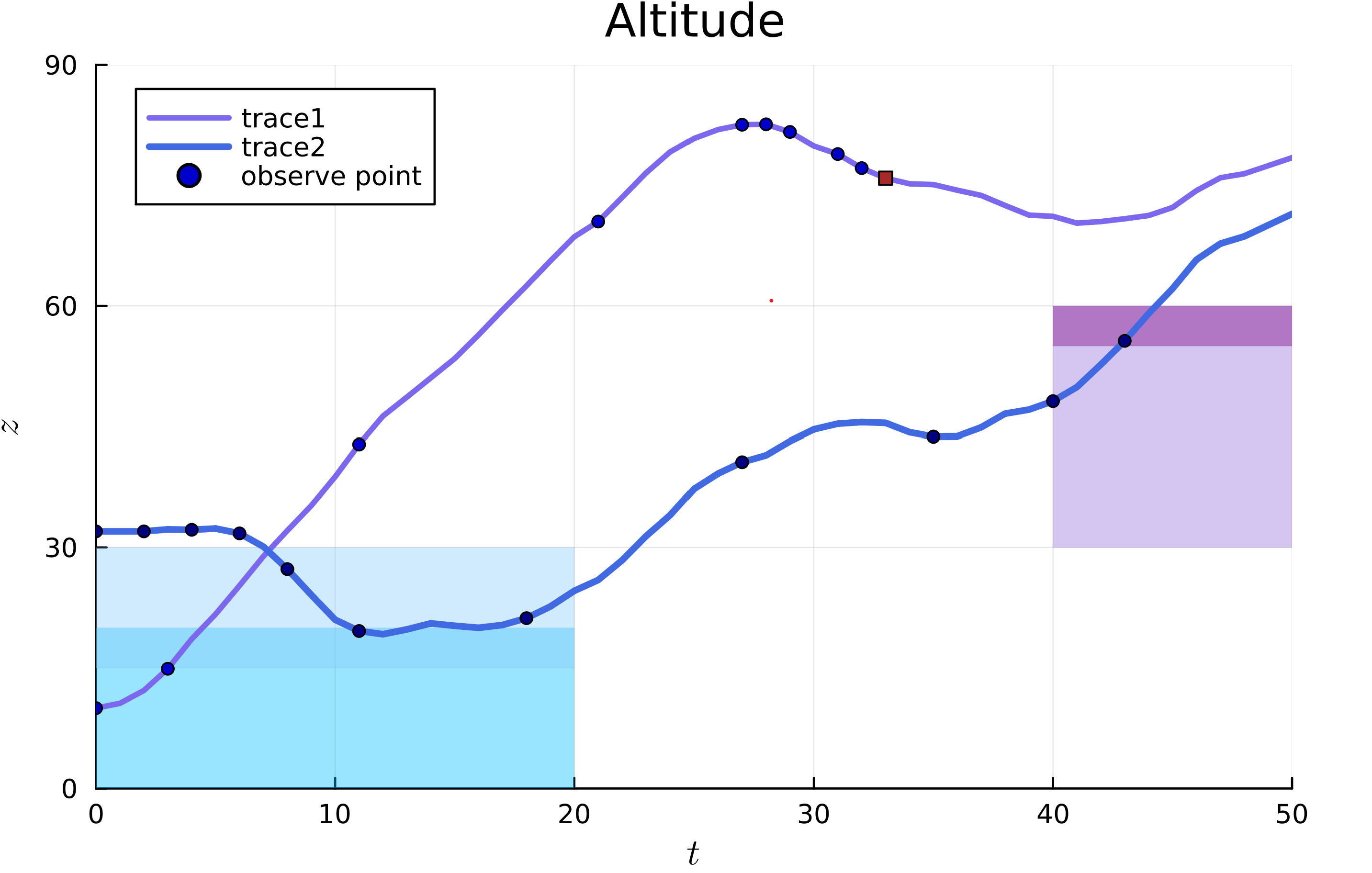}
	}
	\hspace{3mm}
	\subfloat[Left (resp.\ right) hand side is the safe (resp.\ unsafe) belief\\ state predicted for $t=34$ (resp.\ $t=35$) from $t=27$.]
	{\centering
    		\includegraphics[scale=0.23]{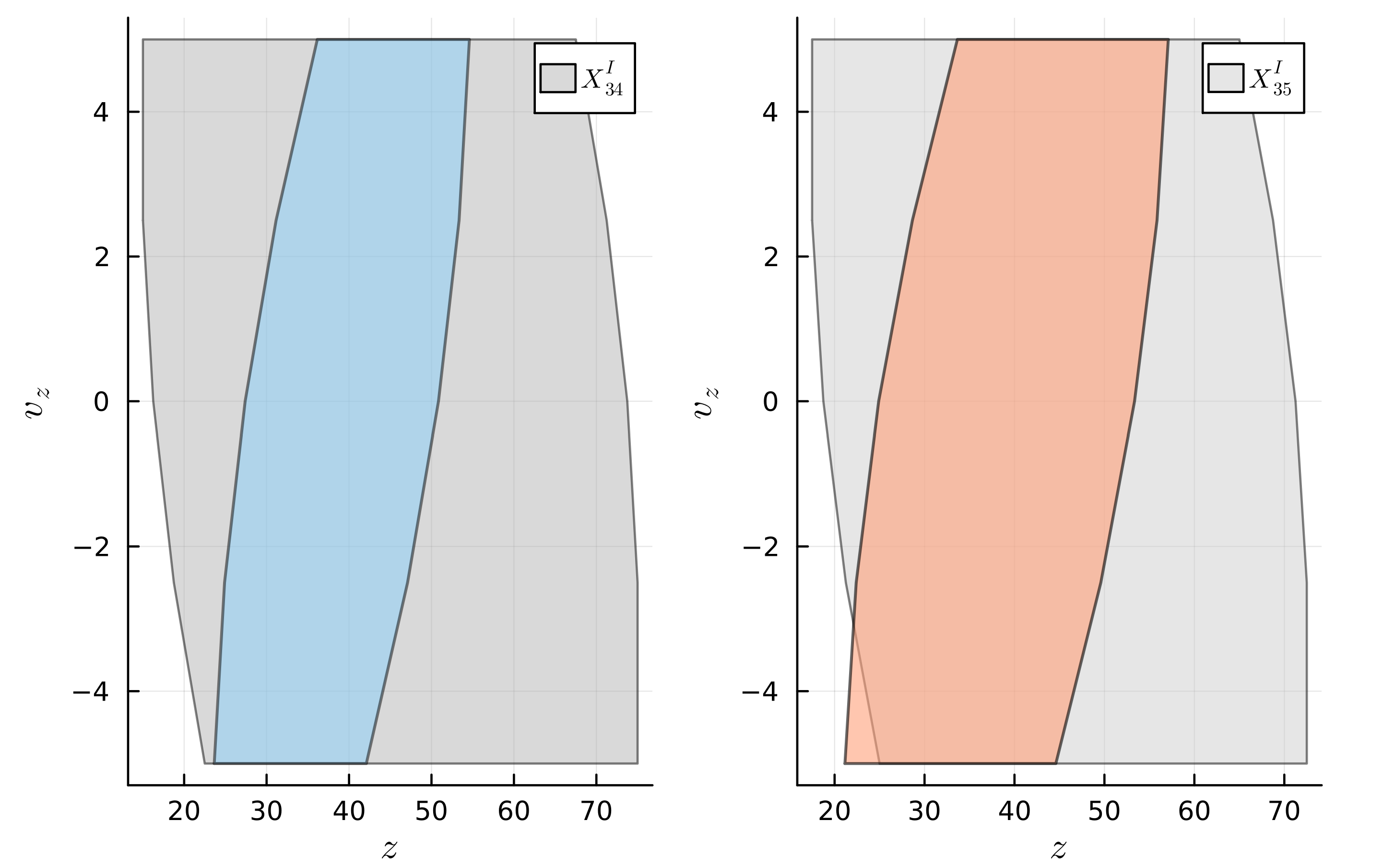}   
	}
	\caption{Online monitoring of drone altitude control systems.} %
	\label{Case1fig1}
\end{figure}

\section{Case Studies}\label{sec-case}
In this section, we demonstrate our approach by case studies.  
Specifically, we have implemented the proposed self-triggered  monitoring algorithm in \textsf{Julia} language \cite{Bezanson2017} and  employ  package \textsf{JuliaReach} \cite{Bogomolov2019} for reachability analysis. 
In particular, for nonlinear systems,  reachable sets   are calculated by over-approximated zonotopes \cite{Althoff2010a}.
Note that, for discrete-time systems in \eqref{system model}, a belief state may contain infinite augmented states since the state space is continuous. Therefore, for the purpose of computation, we represent each belief state $\bm{b}$ equivalently in the form of  $\bm{b}'=\{ (R_{I_1},I_1),\dots, (R_{I_l},I_l) \}$,  
where  $I_i\neq I_j,\forall i\neq j$ and $R_{I}=\{x\in\mathcal{X}\mid (x,I)\in \bm{b} \}$ is the region of all states whose index set is $I$.     
All simulations were performed on a computer with an Intel Core i9-13900H CPU and 32 GB of RAM. 
All codes are available at \url{https://github.com/ChuweiW/st_om}.

\subsection{Case Study 1:  Drone Altitude Control} 
We consider a drone  for the purpose of gathering air data at various altitudes. Therefore, we consider its altitude state $z$ and velocity $v_z$, i.e., $x = [z,v_z]^{\top}\in \mathcal{X}=[0,100]\times [-5,5]$. The discrete-time dynamic   of the drone is 
\begin{equation}
    x_{t+1}=Ax_t+Bu_t
\end{equation}
where $A=
\begin{bmatrix}
1 & 0.5\\
0 & 1
\end{bmatrix}$, $B=
\begin{bmatrix}
0.5\\
1
\end{bmatrix}$
and $\mathcal{U}\subseteq [-2.5,2.5]$. 
In order to initialize the process, the drone needs to first reach  altitude ranges $[0,20]$ and $[15,30]$ in the first $20$ time instants to get the permission to take off. 
Then between $40\sim50$ time instants, it should maintain its altitude between $[30,60]$ until it successfully collects required data at $z\in[55,60]$. We set the time horizon $T=50$ and STL task is 
\begin{equation}
    \begin{aligned}
    \psi &= \mathbf{F}_{[0,20]} z\in[0,20] \wedge  \mathbf{F}_{[0,20]} z\in [15,30]\\ & \wedge  z\in[30,60] \mathbf{U}_{[40,50]} z\in[55,60].
    \end{aligned}
    \label{case1stl}
\end{equation} 

In Fig.~\ref{Case1fig1}(a), we show two possible trajectories of the drone with the starting from different initial altitudes.  
In each trace, blue circles denote those triggered instants at which the monitor takes observations. 
Specifically,  trace~1 fails to achieve the STL task 
and the red square in the trace denotes the instant from which the task is not longer feasible anymore.  
One can see that the monitor takes observations more frequently when the system is approaching the violation of the task. 
On the other hand, trace~2 shows a trajectory along which the STL is achieved successfully.  
Note that, the monitor takes observations at $t=27$ and $t=35$. 
To illustrate procedure \texttt{Trigger-Time}, 
let us focus on the evolution of the predicted belief state from $t=34$ to $t=35$, which is shown in Fig.~\ref{Case1fig1}(b). 
Particularly, we see that the predicted belief state for $t=34$ is still safe, while the predicted belief state for $t=35$ is no longer safe as it may contain state that is not feasible. This is why the monitor has to make an observation at $t=35$.

\subsection{Case Study 2: Spacecraft Rendezvous}
We consider a spacecraft rendezvous problem adopted from \cite{Chan2017}, where a spacecraft \emph{chaser} aims to approach another spacecraft \emph{target} in the same orbital plane satisfying some desired requirements.  
The system model of the \emph{chaser} navigating to the \emph{target} can be described by a two-dimensional nonlinear dynamics as follows:
 \begin{equation}
\begin{aligned}
\ddot{x} &= n^2x + 2n\dot{y} + \frac{\mu}{r^2} - \frac{\mu}{r_c^3}(r+x) + \frac{u_x}{m_c} \\
\ddot{y} &= n^2y - 2n\dot{x} -\frac{\mu}{r_c^3}y + \frac{u_y}{m_c}, \nonumber
\end{aligned}
\end{equation}
 where system state is $x_k = [x\ y\ \dot{x}\ \dot{y}]^{\top}$, control inputs provided by the  thrusters are $u_k = [u_x\ u_y]^{\top}$. The parameters are related to geostationary equatorial orbit with $\mu=3.698\times 10^{14}\times60^2 m^3/min^2$ , $r=42164km$, $m_c = 500kg$, $n=\sqrt{\frac{\mu}{r^3}}$ and $r_c=\sqrt{(r+x)^2 + y^2}$.  
 We discretize the system with sampling time of 0.5  minutes. The physical constraints are $x\in \mathcal{X}=[-100,0]\times[-70,70]\times[0,10]\times[0,10]$ and $u\in \mathcal{U}=[-3,3]\times[-3,3]$.

The formal specifications for the rendezvous system are as follows. 
First, when the distance between \emph{chaser} and \emph{target} is less than 100m, the \emph{chaser} should  keep rendezvousing with relatively stationary velocity and stay within the  line-of-sight  cone  $\textsf{LOS}=\{(x,y)\mid (y\geq x\text{tan}(30^{\circ})\wedge  (-y\geq x\text{tan}(30^{\circ})\}$. 
Furthermore, the rendezvous  task should be completed within 25 minutes in the sense that the \emph{chaser} is close enough to the \emph{target} with velocity constraints. 
We define  $\textsf{Goal}=\{(x,y,\dot{x},\dot{y})\mid x\in[-6,0] \wedge  y\in[-2,2]\wedge  \dot{x}\in[0,3] \wedge  \dot{y}\in[0,3]\}$ 
as the region satisfying the task. 
Finally, during the rendezvous process, the \emph{chaser} should avoid the \textsf{Debris} region represented by the red ball in Fig.~\ref{fig:case2}. 
Therefore, in terms of STL formula, the mission can be  described by:
\[
 \psi = 
 \mathbf{F}_{[0,25]} x_k \!\in\! \textsf{Goal} \wedge  \mathbf{G}_{[0,25]} p \!\notin\! \textsf{Debris} \wedge  \mathbf{G}_{[0,25]} p \!\in\! \textsf{LOS}, 
\]
where $p$ denotes $[x,y]^{\top}$ the position of the \emph{chaser}. 

\begin{figure}[t] 
	\centering
	{\includegraphics[width=8cm]{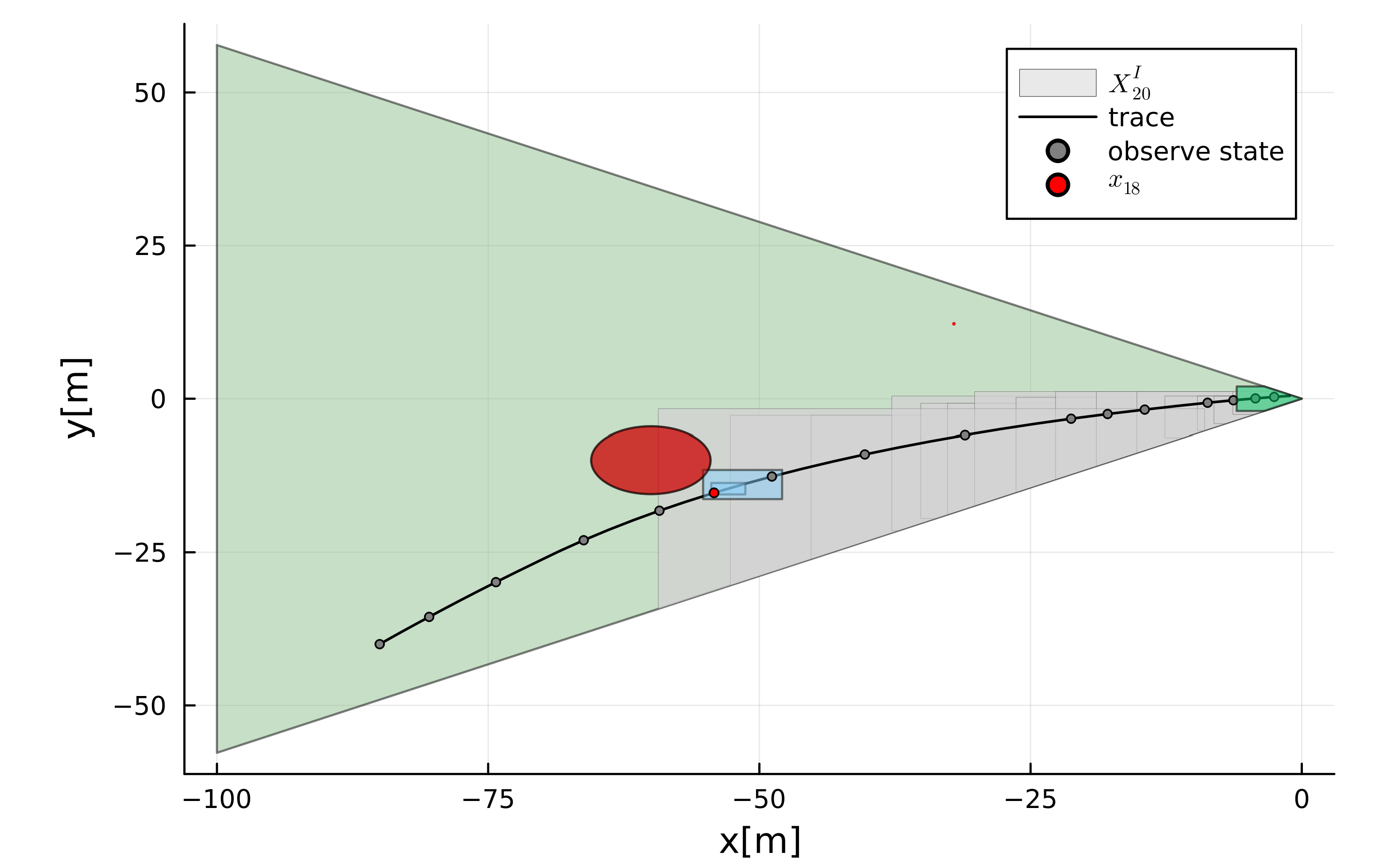}}
	\caption{Online monitoring of spacecraft rendezvous.}
	\label{fig:case2}
\end{figure}

In Fig.~\ref{fig:case2}, we consider a trajectory which completes the task in advance at 40 steps (20min) and the observed states are marked by circles.  Under the  self-triggered algorithm, the monitor only samples system states for 16 times compared with the periodic sampling strategy of 40 times. 
For example,  consider system state $x_{18}$ marked by red circle. 
The blue region shows the evolution of belief state starting from this time instant.
We see that, at time step 20, the predicted belief state intersects with the debris. 
This means that, the spacecraft may hit the \textsf{Debris} region. 
Therefore, the monitor cannot ``sleep'' anymore and has to activate its sensors to solve information ambiguity.  Note that, all feasible sets and reachable sets are four dimensional but are projected into two dimensions (position) for plotting.


\section{Conclusion}\label{sec-con}
In this paper, we introduced a novel \emph{self-triggered} approach for online monitoring of dynamical systems based on signal temporal logic specifications. In contrast to conventional online monitoring methods employing periodic information updates, our proposed self-triggered monitor actively determines sampling instants for system states, offering the potential to significantly reduce sensor energy consumption. The key principle behind our trigger time design is to enable the monitor to remain inactive as long as the predicted belief state exhibits no ambiguity regarding the satisfaction of the STL task.   Illustrative examples are provided to show the effectiveness of the proposed approach.  In the future, we would like to  extend the application of the proposed self-trigger mechanism to address model-based control synthesis problems for STL specifications.

\bibliographystyle{IEEEtran}
\bibliography{reference}
\renewcommand{\theequation}{A.\arabic{equation}}
\setcounter{equation}{0}

\appendix
\begin{myprof}
We define the augmented sequence generated from any prefix $\mathbf{x}_{0:t}$ by the index update function: 
\[
\zeta^{upd}(\mathbf{x}_{0:t})= (x_0,I_0)(x_1,I_1)\ldots(x_t,I_t)
\] 
where $I_{i+1}=\textsf{update}_i(I_{i},x_i),\forall  i>0$. 
Then we prove that for any $\mathbf{x}_{0:t}$, there always exists $\zeta^{upd}(\mathbf{x}_{0:t})=\zeta(\mathbf{x}_{0:t})$.

Initially, we have $\zeta^{upd}(\mathbf{x}_{0:0})\!=\!\zeta(\mathbf{x}_{0:0})\!=\!(x_0,\mathcal{I})$.
Any $I$-remaining formula at time  instant $t$ can be written in the following form of

\begin{equation}\label{I_formula}
\Phi_t^I= \underset{i\in I\cap\mathcal{I}^{\mathbf{U}}_{<t}}{\bigwedge}\Phi_i\underset{i\in I\cap\mathcal{I}^{\mathbf{U}}_{t}}{\bigwedge}\Phi_i\underset{i\in \mathcal{I}_{t}^{\mathbf{G}}}{\bigwedge}\Phi_i\underbrace{\underset{i\in \mathcal{I}_{>t}^{\mathbf{U}}}{\bigwedge}\Phi_i\underset{i\in \mathcal{I}_{>t}^{\mathbf{G}}}{\bigwedge}\Phi_i}_{\hat{\Phi}^I_{>t}}
\end{equation}
Obviously, once the STL task is settled, the indices of ``always'' temporal operator are fixed and irrelevant with system evolution. Then it suffices to focus on the changes in indices of ``until'' temporal operator.

For the $I$-remaining index in $\zeta(\mathbf{x}_{0:t})$, we know that the index of ``until'' operator will be contained in $I_t$ if it is not satisfied by prefix $\mathbf{x}_{0:t-1}$. 

For the $I$-remaining formula in $\zeta^{upd}(\mathbf{x}_{0:t})$, we consider the augmented state at time $t$, i.e. $(x_t,I'_t)$, where the prime is used to distinguish from the remaining formula in augmented state in $\zeta(\mathbf{x}_{0:t})$. Consider the index that is not in $I'_t$. At the beginning, $I_0$ contains all the indices, the ``until'' temporal operator will be removed if satisfied, and the ``always'' temporal operator will be removed if expired. Combing the index update function, $I'_t$ can be written as follows
\begin{equation}
    \Phi^{I'_t}_t\!=\!\! \underset{\substack{i\in\mathcal{I}^{\mathbf{U}}_{<t},\nexists k\in[a_i,b_i]\\ s.t. x_k\in H_i^1\cap H_i^2}}{\bigwedge}\!\!\Phi^{[a_i,b_i]}_i\underset{\substack{i\in\mathcal{I}^{\mathbf{U}}_{t},\nexists k\in[a_i,t)\\ s.t. x_k\in H_i^1\cap H_i^2}}{\bigwedge}\!\!\Phi^{[a_i,t]}_i\underset{i\in \mathcal{I}_{t}^{\mathbf{G}}}{\bigwedge}\!\Phi_i\underset{i\in \mathcal{I}_{>t}}{\bigwedge}\!\Phi_i
    \label{updateI}
\end{equation}
It follows that the indices contained in $I'_t$ are exactly the same as $I_t$.
Then we complete the proof that $\zeta^{upd}(\mathbf{x}_{0:t})=\zeta(\mathbf{x}_{0:t})$ always holds. Furthermore, we can say that starting from the same augmented state $(x_k,I)$, $\zeta^{upd}_{I}(\mathbf{x}_{k:t})=\zeta_{I}(\mathbf{x}_{k:t})$ is true where $I$ denotes the remaining formula paired $x_k$.

Next, we complete the proof by induction on observation times $n$ as follows:

\emph{Base case:} For $n\!=\!0$, we have $\hbar\!=\!(x_0,\tau_0)$ and $\bm{\hat{b}}_0\!=\!\bm{b}_0\!=\!\{(x_0,\mathcal{I})\}$. $\mathcal{M}_{\bm{b}}^{-1}(\hbar)$ can only be $\{(x_0,\mathcal{I})\}$ to be consistent with observation history. Then we know that $\bm{b}_0=\textsf{last}(\mathcal{M}_{\bm{b}}^{-1}(\hbar))$.

\emph{Induction Step:} Assume that at times of observation $n$, i.e., $\hbar\!=\!(x_0,\tau_0)\ldots (x'_{n-1},\tau_{n-1})x'_{n}$ and $\bm{b}'_n = \{(x'_n,I_n^1),\ldots,(x'_n,I_n^m) \}$, we have $\bm{b}'_n=\textsf{last}(\mathcal{M}_{\bm{b}}^{-1}(\hbar))$.

At $n+1$ observation times, the history updates to
\[
\hbar'=(x_0,\tau_0)\ldots (x'_n,\tau_{n})x'_{n+1}
\]
Thus, the next predicted belief state is given by
\[\hat{\bm{b}}'_{n+1}=\underset{(x'_n,I_n)\in\bm{b}'_n}{\cup}\textsf{last}(\zeta^{upd}_{I_n}(\mathbf{x}_{T_{n-1}:T_{n}}))\] 
Under the assumption that $\bm{b}'_n=\textsf{last}(\mathcal{M}_{\bm{b}}^{-1}(\hbar))$, we have 
\[
\mathcal{M}_{\bm{b}}^{-1}(\hbar')\!=\! \{\textsf{last}(\zeta_{I_n}(\mathbf{x}_{T_{n-1}:T_{n}})\mid (x'_n,I_n)\!\in\!\bm{b}'_n , x_{T_{n}}\!=\!x'_{n+1} \}
\]
Due to the fact that $\zeta^{upd}_{I}(\mathbf{x}_{k:t})\!=\!\zeta_{I} (\mathbf{x}_{k:t})$ we have proven before, it holds that
$\bm{b}'_{n+1}\!=\!\textsf{refine}(\bm{\hat{b}}'_{n+1},x'_{n+1})\!=\!\mathcal{M}_{\bm{b}}^{-1}(\hbar')$.
The proof is thus completed.
\end{myprof}
\end{document}